\documentclass[runningheads]{llncs}

\usepackage{amsmath}
\usepackage{array}
\usepackage{ascmac}
\usepackage{amssymb}
\usepackage{booktabs}
\usepackage[dvipdfmx]{graphicx}
\usepackage{wrapfig}
\usepackage[algoruled,linesnumbered,algo2e,vlined]{algorithm2e}
\usepackage{url}
\usepackage{multirow}
\usepackage{multicol}
\usepackage{cases}
\usepackage[normalem]{ulem}
\usepackage{boites}
\usepackage{hyperref}
\usepackage{svg}

\usepackage[final,mode=multiuser]{fixme}

\usepackage{colortbl}

\fxsetup{theme=color}
\definecolor{fxtarget}{rgb}{0.0000,0.0000,0.6823}

\fxuseenvlayout{colorsig}
\fxusetargetlayout{color}

\fxuseenvlayout{colorsig}
\fxusetargetlayout{color}
\FXRegisterAuthor{yf}{ayf}{YF}
\FXRegisterAuthor{hs}{ahs}{HS}
\FXRegisterAuthor{yn}{ayn}{YN}
\FXRegisterAuthor{si}{asi}{SI}
\fxsetface{env}{}

\usepackage{orcidlink}
\renewcommand{\orcidID}[1]{\orcidlink{#1}}

\newcommand{\AS}[1]{\mathsf{AS}_{{#1}}}
\newcommand{\MS}[1]{\mathsf{MS}_{{#1}}}

\newcommand{\ed}{\mathsf{ed}}

\newcommand{\edit}{\mathrm{edit}}
\newcommand{\ins}{\mathrm{ins}}
\newcommand{\del}{\mathrm{del}}
\newcommand{\sub}{\mathrm{sub}}

\newcommand{\bs}{\mathit{b}}

\newcommand{\lzsssr}{\mathsf{LZSS}}
\newcommand{\LZsssr}{\mathit{z}_{\mathrm{ss}}}

\newcommand{\lzss}{\mathsf{LZSS}_{no}}
\newcommand{\LZss}{\mathit{z}_{\mathrm{ssno}}}

\newcommand{\lzend}{\mathsf{LZEnd}}
\newcommand{\LZEnd}{\mathit{z}_{\mathrm{e}}}
\newcommand{\LZEndOpt}{\mathit{z}_{\mathrm{end}}}

\newcommand{\lzseveneight}{\mathsf{LZ78}}
\newcommand{\LZSevenEight}{\mathit{z}_{\mathrm{78}}}

\newcommand{\LZ}{\mathit{z}_{\mathrm{77no}}}
\newcommand{\LZsr}{\mathit{z}_{\mathrm{77}}}

\newcommand{\Substr}{\mathsf{Substr}}

\begin{document}

\title{Tight Additive Sensitivity \\ on LZ-style Compressors and String Attractors}

\author{%
Yuto~Fujie\inst{1} %
\and
Hiroki~Shibata\inst{1}\orcidID{0009-0006-6502-7476}
\and
Yuto~Nakashima\inst{2}\orcidID{0000-0001-6269-9353}
\and
Shunsuke~Inenaga\inst{2}\orcidID{0000-0002-1833-010X}
}

\institute{%
    Joint Graduate School of Mathematics for Innovation, Kyushu University, Japan\\
    \email{\{fujie.yuto.104, shibata.hiroki.753\}@s.kyushu-u.ac.jp}
    \and
    {Department of Informatics, Kyushu University, Japan}\\
    \email{\{nakashima.yuto.003, inenaga.shunsuke.380\}@m.kyushu-u.ac.jp}
}

\authorrunning{Fujie, Shibata, Nakashima, Inenaga}

\maketitle
\begin{abstract}
  The \emph{worst-case additive sensitivity} of a string repetitiveness measure $c$ is defined to be the largest difference between $c(w)$ and $c(w')$,
  where $w$ is a string of length $n$ and $w'$ is a string that can be obtained by performing a single-character edit operation on $w$.
  We present $O(\sqrt{n})$ upper bounds for the worst-case additive sensitivity of the smallest string attractor size $\gamma$ and the smallest bidirectional scheme size $\bs$, which match the known lower bounds $\Omega(\sqrt{n})$ for $\gamma$ and $\bs$~[Akagi et al. 2023].
  Further, we present matching upper and lower bounds 
  for the worst-case additive sensitivity of the Lempel-Ziv family - 
  $\Theta(n^{\frac{2}{3}})$ for LZSS and LZ-End, and $\Theta(n)$ for LZ78.
\end{abstract}


\section{Introduction}

Measuring the \emph{repetitiveness} of strings
is one of the most fundamental studies on strings which attract recent attention.
Examples of string repetitiveness measures are
the \emph{substring complexity} $\delta$~\cite{KociumakaNP20},
the smallest \emph{string attractor} size $\gamma$~\cite{KempaP18},
the number $r$ of runs in the \emph{Burrows-Wheeler transform}~\cite{Burrows94ablock-sorting},
the size $\bs$ of the smallest \emph{bidirectional scheme}~\cite{StorerS82},
the size $z$ of the Lempel-Ziv factorization~\cite{LZ77} (and its variants~\cite{StorerS82,KreftN13,KempaS22}), and
the smallest grammar size $g$~\cite{CharikarLLPPSS05,Rytter03}.
We refer readers to the survey~\cite{repetitiveness_Navarro21a} for comparisons of these measures and others.

Akagi et al.~\cite{AKAGI2023104999} introduced the notion of \emph{sensitivity} of string repetitiveness measures, which evaluates the perturbation of the measures after an edit operation (insertion, deletion, or substitution of a character) is performed on the input string.
Sensitivity is a mathematical model for robustness of string compressors/repetitiveness measures against dynamic changes and/or errors.
Of two versions of sensitivity,
the work in~\cite{AKAGI2023104999} mostly focuses on
the worst-case \emph{multiplicative sensitivity} $\MS{\edit}(c,n)$
with edit operation $\edit \in \{\ins, \del, \sub\}$
that is defined to be the maximum of $c(T')/c(T)$, where $c$ is the measure, $T$ is any string of length $n$, and $T'$ is any string that can be obtained by a single-character edit operation from $T$.
The other alternative is the worst-case \emph{additive sensitivity} $\AS{\edit}(c,n)$, which is defined to be the maximum of $c(T') - c(T)$.
We note that most of the results on the additive sensitivity in~\cite{AKAGI2023104999}
are byproducts of their multiplicative sensitivity,
and they are evaluated in terms of the measure $c$.
For instance, the multiplicative sensitivity $3$ of the overlapping LZSS $\LZsssr$ immediately leads to a $2 \LZsssr$ additive sensitivity~\cite{AKAGI2023104999}.
The behavior of $\AS{\edit}(c,n)$ in terms of $n$ is not well understood
for a majority of measures,
except for the trivial bound for $\delta$~\cite{AKAGI2023104999} and
the $\Omega(\sqrt{n})$ lower bounds for $\gamma$, $\bs$ and $r$~\cite{AKAGI2023104999,GiulianiILRSU25}.

In this paper, we show tight upper and lower bounds of $\AS{\edit}(c,n)$ in terms of $n$ for the following measures $c$:
the smallest string attractor $\gamma$,
the smallest bidirectional scheme $\bs$,
the Lempel-Ziv family including
the overlapping/non-overlapping LZSS $\LZsssr$ and $\LZss$~\cite{StorerS82},
the optimal LZ-End $\LZEndOpt$~\cite{KempaS22},
and LZ78 $\LZSevenEight$~\cite{LZ78}.
We present $O(\sqrt{n})$ upper bounds for the worst-case additive sensitivity of $\gamma$ and $\bs$, which match the known lower bounds $\Omega(\sqrt{n})$ for $\gamma$ and $\bs$~\cite{AKAGI2023104999}.
Further, we present matching upper and lower bounds $\Theta(n^{\frac{2}{3}})$ 
for the worst-case additive sensitivity of $\LZsssr$, $\LZss$, and $\LZEndOpt$.
We also present an $\Omega(n)$ lower bound instance for
the worst-case additive sensitivity of $\LZSevenEight$,
which matches a na\"ive $O(n)$ upper bound.
Our results and previous results are summarized in Table~\ref{tbl:comparisons}.

Independently to our work,
Blocki et al.~\cite{BlockiLSY25} presented
an $O(n^{\frac{2}{3}})$ upper bound
and an $\Omega(n^{\frac{2}{3}} \log^{\frac{1}{3}} n)$ lower bound
for the worst-case additive sensitivity of the non-overlapping LZ77 factorizations~\cite{LZ77}.
While they provide a specific constant factor for the upper bound,
their lower bound is slightly loose (up to a logarithmic factor).
We remark that our tight $\Theta(n^{\frac{2}{3}})$ sensitivity bounds
for (non)overlapping LZSS can readily extend to (non)overlapping LZ77.

\begin{table}[t]
  \centering
  \label{tbl:comparisons}
  \caption{The worst-case additive sensitivity of the string compressors/repetitiveness measures studied in this paper, evaluated in terms of the length $n$ of the string. All the bounds hold for any type of edit operations (substitutions, insertions, and deletions). All the bounds except those from~\cite{AKAGI2023104999,GiulianiILRSU25} are our new results.}
  \begin{tabular}{|l||c|c|}
    \hline
    \multirow{2}{*}{string compressor/repetitiveness measure} & \multicolumn{2}{c|}{worst-case additive sensitivity} \\ \cline{2-3}
    & upper bound & lower bound \\ \hline \hline
    substring complexity $\delta$ & 1~\cite{AKAGI2023104999} & 1~\cite{AKAGI2023104999} \\ \hline
    smallest string attractor $\gamma$ & $O(\sqrt{n})$ & $\Omega(\sqrt{n})$~\cite{AKAGI2023104999} \\ \hline
    smallest bidirectional scheme $\bs$ & $O(\sqrt{n})$ & $\Omega(\sqrt{n})$~\cite{AKAGI2023104999} \\ \hline
    non-overlapping LZSS $\LZss$ & $O(n^{\frac{2}{3}})$ & $\Omega(n^{\frac{2}{3}})$ \\ \hline
    overlapping LZSS $\LZsssr$ & $O(n^{\frac{2}{3}})$ & $\Omega(n^{\frac{2}{3}})$ \\ \hline
    non-overlapping LZ77 $\LZ$ & $O(n^{\frac{2}{3}})$~\cite{BlockiLSY25} & $\Omega(n^{\frac{2}{3}})$ \\ \hline
    overlapping LZ77 $\LZsr$ & $O(n^{\frac{2}{3}})$ & $\Omega(n^{\frac{2}{3}})$ \\ \hline
    optimal LZ-End $\LZEndOpt$ & $O(n^{\frac{2}{3}})$ & $\Omega(n^{\frac{2}{3}})$ \\ \hline
    greedy LZ-End $\LZEnd$ & - & $\Omega(n^{\frac{2}{3}})$ \\ \hline 
    LZ78 $\LZSevenEight$ & $O(n)$ & $\Omega(n)$ \\ \hline
    run-length BWT $r$ & - & $\Omega(\sqrt{n})$~\cite{GiulianiILRSU25} \\ \hline
  \end{tabular}
\end{table}

\section{Preliminaries}

\subsection{Strings}
Let $\Sigma$ be an {\em alphabet}.
An element of $\Sigma^*$ is called a {\em string}.
The length of a string $T$ is denoted by $|T|$.
The empty string $\varepsilon$ is the string of length 0.
Let $\Sigma^n$ denote the set of strings of length $n$.
For a string $T = xyz$, $x$, $y$ and $z$ are called
a \emph{prefix}, \emph{substring}, and \emph{suffix} of $T$, respectively.
Let $\Substr(T)$ denote the set of substrings of $s$.
The $i$-th symbol of a string $T$ is denoted by $T[i]$ for $1 \leq i \leq |T|$.
Let $T[i..j]$ denote the substring of $T$ that begins at position $i$ and ends at position $j$ for $1 \leq i \leq j \leq |T|$.
For convenience, let $T[i..j] = \varepsilon$ when $i > j$.
For a sequence $s_1, \ldots, s_k$ of strings,
let $\prod_{i = 1}^ks_k = s_1 \cdots s_k$ denote their concatenations.

\subsection{String attractors}

Let $T$ be a non-empty string of length $n$.
A set $\Gamma = \{p_1, \ldots, p_{\ell}\}$ of $\ell$ positions in $T$
is said to be a \emph{string attractor}~\cite{KempaP18} of $T$
if any substring $s \in \Substr(T)$ has an occurrence $[i,j]$
such that $s = T[i..j]$ and $i \leq p_k \leq j$ for some $p_k \in \Gamma$.
Since the set $\{1, \ldots, n\}$ of all positions in $T$ clearly satisfies
the definition, any string has a string attractor.
The size of a string attractor is the number of positions in it.
Let $\gamma(T)$ denote the size of the smallest string attractor(s) for $T$.
For instance, the set $\{ 5,7 \}$ of positions is a string attractor of string $T= \mathtt{baaaabbaaa}$
(see also Fig.~\ref{fig:examples}),
and it is the smallest since $T$ contains two distinct characters $\mathtt{a,b}$.

\subsection{Bidirectional macro scheme}
A sequence $f_1, \ldots, f_b$ of non-empty substrings of a string $T$
is said to be a \emph{bidirectional macro scheme} of $T$
if $T = f_1 \cdots f_b$ and each $f_k = T[p_k,p_k+|f_k|-1]$ (called a phrase) is either a single character 
or a copy of a substring $T[q_k,q_k+|f_k|-1]$ (called a source) occurring to the left or right of $f_k$.
If $|f_k| = 1$, then $f_k$ is called a ground phrase.
The function $F:[1,n] \rightarrow [1,n] \cup \{ 0\}$ such that
\begin{alignat*}{2}
  F(0) &= 0 & & \\
  F(p_k) &= 0 & & \text{ if $f_k$ is a grand phrase}, \\
  F(p_k+j) &= q_k + j & & \text{ if $|f_k| > 1$ and $0 \leq j < |f_k|$},
\end{alignat*}
is induced by the bidirectional macro scheme.
Let $F^0(p_k) = p_k $ and $F^m(p_k) = F(F^{m-1}(p_k))$ for $m \leq 1$.
A bidirectional macro scheme $B$ is called valid if there exists an $m \geq 1$ such that $F^m(i)=0$
for every $i \in [1,n]$. The string $T$ can be reconstructed from 
the bidirectional macro scheme if and only if it is valid.
The size of a bidirectional macro scheme $B$ is the number of phrases of $B$.
Let $b(T)$ denote the size of the smallest valid bidirectional macro scheme(s) for $T$.

\subsection{LZ-style compressors}

A sequence $f_1, \ldots, f_z$ of non-empty substrings of a string $T$
is said to be an \emph{LZ-style factorization} of $T$
if $T = f_1 \cdots f_z$ and each $f_k$ is either a fresh character
not occurring to its left, or a copy of a previous substring occurring to its left.

The \emph{overlapping LZSS} factorization of a string $T$
is the greedy LZ-style factorization of $T$
such that each $f_k$ is taken as long as possible.
Hence, each copied phrase $f_k$ in the overlapping LZSS factorization
is the longest prefix of $T[|f_1 \cdots f_{k-1}|+1..n]$
that occurs at least twice in $T[1..|f_1 \cdots f_k|]$.
Let $\lzsssr(T)$ denote the overlapping LZSS factorization of a string $T$
and $\LZsssr(T)$ the number of phrases in $\lzsssr(T)$.

In the \emph{non-overlapping LZSS} factorization,
each copied phrase $f_k$ is the longest prefix of $T[|f_1 \cdots f_{k-1}|+1..n]$
that occurs at least once in $T[1..|f_1 \cdots f_{k-1}|]$.
Let $\lzss(T)$ denote the non-overlapping LZSS factorization of a string $T$
and $\LZss(T)$ the number of phrases in $\lzss(T)$.

An LZ-style factorization $f_1, \ldots, f_z$ of a string $T$ is said to be
an \emph{LZ-End factorization} of $T$ if each copied phrase $f_k$
has a previous occurrence $[i..j]$ in $T[1..|f_1 \cdots f_{k-1}|]$
such that $j$ is the ending position of a previous phrase,
namely $j = |f_1 \cdots f_h|$ for some $1 \leq h < k$.
The \emph{optimal LZ-End factorization} is the LZ-End factorization
whose size is the smallest. Let $\LZEndOpt$ denote the size of 
the optimal LZ-End factorization of string $T$.
The \emph{greedy LZ-End factorization} is the LZ-End factorization
such that each copied phrase $f_k$ is taken as long as possible,
namely, $f_k$ is the longest prefix of $T[|f_1 \cdots f_{k-1}|+1..n]$
that is a suffix of $T[1..|f_1 \cdots f_h|]$ for some $1 \leq h < k$.
Let $\lzend(T)$ denote the greedy LZ-End factorization of a string $T$,
and $\LZEnd(T)$ the number of phrases in $\lzend(T)$.

An LZ-style factorization $f_1, \ldots, f_z$ of a string $T$ is said to be
an \emph{LZ78 factorization} of $T$
if each copied phrase $f_k$ is the shortest prefix of $T[|f_1 \cdots f_{k-1}|+1..n]$ such that $f_k[1..|f_k|-1] = f_h$ for some $1 \leq h < k$ and $f_k \neq f_\ell$ for any $1 \leq \ell < k$.
Let $\lzseveneight(T)$ denote the LZ78 factorization of a string $T$,
and $\LZSevenEight(T)$ the number of phrases in $\lzseveneight(T)$.

Fig.~\ref{fig:examples} shows concrete examples of the aforementioned
string compressors and repetitiveness measures.
\begin{figure}[t]
  \centering
  \includegraphics[keepaspectratio,width = 12cm]{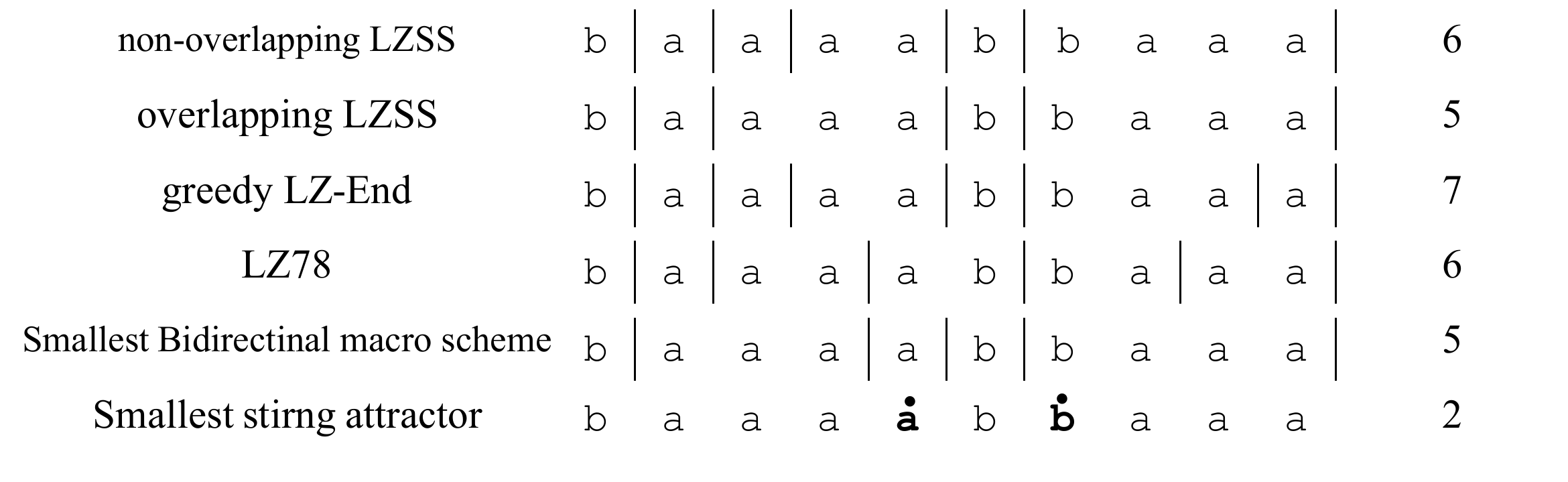}
  \caption{Examples of string compressors/repetitiveness measures for string $T = \mathtt{baaabbaaa}$.}
  \label{fig:examples}
\end{figure}

\subsection{Additive sensitivity of string repetitiveness measures}  
For two strings $S, T$, let $\ed(S,T)$ denote the edit distance
between $S$ and $T$.

Let $c(T)$ denote the size of the output of compressor (such as LZSS) or
a repetitiveness measure (such as the smallest string attractor) for a string $T$ of length $n$.
The worst-case \emph{additive sensitivity},\emph{multiplicative sensitivity} of $c$ is defined by
\begin{eqnarray*}
  \AS{\edit}(c, n) & = \max_{T \in \Sigma^n, T \in \Sigma^m}\{c(T') - c(T) \mid \ed(T') = 1\}, \\
  \MS{\edit}(c, n) & = \max_{T \in \Sigma^n, T \in \Sigma^m}\{c(T') / c(T) \mid \ed(T') = 1\}
\end{eqnarray*}
where $\edit \in \{\sub, \ins, \del\}$
represents the type of the edit operation (substitution, insertion, deletion),
and $m = n$ for $\edit = \sub$,
$m = n+1$ for $\edit = \ins$, and
$m = n-1$ for $\edit = \del$.

\section{Additive Sensitivity of String Attractors}

In this section, we consider the additive sensitivity of
the smallest string attractor size $\gamma$.
The following lower bound is known:
\begin{theorem}[\cite{AKAGI2023104999}] \label{thm:SA_lb}
  $\AS{\edit}(\gamma,n) = \Omega(\sqrt{n})$ holds for $\edit \in \{\sub, \ins, \del\}$.
\end{theorem}

Together with Theorem~\ref{thm:string-attractor},
we obtain a tight bound $\Theta(\sqrt{n})$ for $\AS{\edit}(\gamma, n)$.

\begin{theorem} \label{thm:string-attractor}
  $\AS{\edit}(\gamma,n) = O(\sqrt{n})$ holds for $\edit \in \{\sub, \ins, \del\}$.
\end{theorem}

For ease of discussion,
we consider the case of substitution where $\edit = \sub$.
Let $T$ be the original string of length $n$ and 
$T'$ be the string obtained by replacing the $i$-th character $T[i]=a$ with
a distinct character $b~(\neq a)$. 
Namely, $T' = T[1..i-1] \cdot b \cdot T[i+1..n]$.
We consider three subsets $S_1$, $S_2$, $S_3$ of $\Substr(T')$ as follows:
\begin{align*}
  S_1 &= \{ s \in \Substr(T') \mid \text{$s$ has an occurrence in $T$ that contains position $i$} \}, \\
  S_2 &= \{ s \in \Substr(T') \mid \text{$s$ has an occurrence in $T'$ that contains position $i$} \}, \\
  S_3 &= \Substr(T') - (S_1 \cup S_2).
\end{align*}

\begin{lemma} \label{lem:attractor-S1}
  There exists an integer set $A \subseteq [1,n]$ 
  such that $|A| = O(\sqrt{n})$ holds and every $s \in S_1$ has an occurrence in $T'$ that contains some $p \in A$.
\end{lemma}

\begin{proof}
  We consider two cases w.r.t. the length of $s \in S_1$.
  \begin{enumerate}
    \item Assume that $|s| \geq 1 + \lceil \sqrt{n} \rceil $.
      Let 
      \[
        A_1 = \left\{ \lfloor\sqrt{n} \rfloor , 2 \lfloor \sqrt{n} \rfloor, \dots, \lfloor\sqrt{n} \rfloor*\lfloor\sqrt{n} \rfloor \right\} 
        \cup \left\{ \left\lfloor \frac{{\lfloor \sqrt{n} \rfloor}^2 + n}{2} \right\rfloor \right\}.
      \]
      We can see that $n - \left\lfloor \frac{{\lfloor \sqrt{n} \rfloor}^2 + n}{2} \right\rfloor \leq 1 + \sqrt{n}$,
      $\left\lfloor \frac{{\lfloor \sqrt{n} \rfloor}^2 + n}{2} \right\rfloor - {\lfloor\sqrt{n} \rfloor}^2 \leq 1 + \sqrt{n}$, and
      $k \lfloor \sqrt{n} \rfloor -(k-1)\lfloor \sqrt{n} \rfloor = \lfloor \sqrt{n} \rfloor \leq \sqrt{n}$ for every $k \in [1,\lfloor \sqrt{n} \rfloor]$.
      These facts imply that every $s$ has an occurrence in $T$ that contains some $p \in A_1$.
      It is clear that $|A_1| \leq \lfloor \sqrt{n} \rfloor + 1= O(\sqrt{n})$ holds.

    \item Assume that $|s| \leq \lceil \sqrt{n} \rceil$.
      We consider a subset $S_1' \subseteq S_1$ 
      such that $|s| \leq \lceil \sqrt{n} \rceil$ holds and the occurrence (as an interval) of $s$ is not nested 
      by an occurrence of any other elements in $S_1'$.
      From such structures, we can see that $|S_1'| \leq \lceil \sqrt{n} \rceil$.
      Let $s \in S_1'$ and $\ell$ be the position in $s$ that corresponds to the edit position.
      We choose an occurrence $j$ of $s$ in $T'$ and also consider a position $j+\ell-1$.
      We can see that this position stabs $s$ and also shorter substrings $s' \notin S_1'$.
      Thus a set $A_2$ of such positions for every $s \in S_1'$ covers an occurrence of each substring in this case.
      It is clear that $|A_2| \leq \lceil \sqrt{n} \rceil = O(\sqrt{n})$ holds.
    \end{enumerate}
    From the above discussion, $A = A_1 \cup A_2$ satisfies the statement.
\end{proof}

We show the following lemma.

\begin{lemma} \label{lem:construct-gamma}
  For any strings $T$ of length $n$, 
  there exists a string attractor of $T'$ of size $\gamma(T)+O(\sqrt{n})$.
\end{lemma}

\begin{proof}
  Let $\Gamma$ be a string attractor of $T$.
  We show that there exists a string attractor $\Gamma'$ of $T'$ whose size is $|\Gamma|+O(\sqrt{n})$.
  By Lemma~\ref{lem:attractor-S1}, positions in $A$ that is defined in Lemma~\ref{lem:attractor-S1} cover substrings in $S_1$.
  It is easy to see that position $i$ stabs all substrings in $S_2$.
  For any substring in $S_3$, it has an occurrence in both strings that contain a position in $\Gamma$.
  Thus $\Gamma$ covers substrings in $S_3$.
  Overall, there exists a string attractor $\Gamma' = \Gamma \cup A \cup \{i\}$ of size $|\Gamma|+O(\sqrt{n})$.
\end{proof}

Lemma~\ref{lem:construct-gamma} immediately derives Theorem~\ref{thm:string-attractor}.

\section{Additive Sensitivity of Bidirectional Macro Scheme}
In this section, we consider the additive sensitivity of
the smallest bidirectional macro scheme size $b$.
Akagi et al.~\cite{AKAGI2023104999} showed the following lower bounds for each edit operation.
\begin{theorem}[\cite{AKAGI2023104999}] \label{BMS_lb}
  $\AS{\edit}(b,n) = \Omega(\sqrt{n})$ holds for $\edit \in \{\sub, \ins, \del\}$.
\end{theorem}

Together with Theorem~\ref{bms_ub},
we obtain a tight bound $\Theta(\sqrt{n})$ for $\AS{\edit}(b, n)$.
\begin{theorem} \label{bms_ub}
  $\AS{\edit}(b,n) = O(\sqrt{n})$ holds for $\edit \in \{\sub, \ins, \del\}$.
\end{theorem}

To prove Theorem~\ref{bms_ub}, we use the proof of the multiplicative sensitivity of 
bidirectional macro scheme by Akagi et al \cite{AKAGI2023104999}.
Let $T$ be a string and $T'$ the string obtained by editing the $i$-th character of $T$,
and $f_1, \ldots, f_b$ be a valid bidirectional macro scheme of $T$. 
They showed how to construct a valid bidirectional macro scheme of $T'$ from $f_1, \ldots, f_b$
and the increasing number of phrases in the operations.
We denote the beginning position of phrase $f_j$ by $p_j$ (i.e., $f_j = T[p_j..p_j+|f_j|-1]$)
and its beginning position of source by $q_j$ (i.e., $T[q_j..q_j+|f_j|-1]$).
Akagi et al. considered the three types of operations.
\begin{proposition}[Theorem~{11} of \cite{AKAGI2023104999}] \label{bms_ms_ub}
  A valid bidirectional macro scheme of $T'$ can be obtained 
  by applying the following operations to $f_1, \ldots, f_b$.
  \begin{itemize}
    \item[(1)] If $i \in [p_j,p_j+|f_j|-1]$, we divide the phrase $f_j$ into at most five phrases.  
    \item[(2)] If $i \notin [p_j,p_j+|f_j|-1]$ and $i \notin [q_j,q_j+|f_j|-1]$, we don't need any changes.
    \item[(3)] If $i \notin [p_j,p_j+|f_j|-1]$ and $i \in [q_j,q_j+|f_j|-1]$, 
    we divide the phrase $f_j$ at most three phrases.
  \end{itemize} 
\end{proposition}
In case (1), the number of phrases increases by at most two, and in case (2), the number of factor remains unchanged.
Here, we give an upper bound of increasing number of phrases in case (3) w.r.t. $n$.
First, we consider the number of phrases in case (3) whose length is at least $\sqrt{n}$.
Then, such phrases can occur at most $O(\sqrt{n})$ times in $T$, since $|T|= n$.
Hence the total increasing number of phrases from such phrases is $O(\sqrt{n})$, 
since each phrase is divided into at most three phrases by Proposition~\ref{bms_ms_ub}.
Next, we consider the number of phrases whose length is smaller than $\sqrt{n}$.
Let $f_k$ be a phrase in case (3) such that there exists a phrase $f_{k'}$ that is a superstring of $f_k$ ($k' \neq k$).
For such phrase $f_k$, a corresponding substring of $f_{k'}$ can be a new source of $f_k$ 
even if the original source of $f_k$ is lost by editing.
Thus, in this case, it is unnecessary to divide $f_k$.
On the other hand, the maximum number of phrases in case (3) whose sources do not nest each other is $O(\sqrt{n})$
since the sources contain the edited position $i$.
For such $O(\sqrt{n})$ phrases, we divide them into at most three phrases by Proposition~\ref{bms_ms_ub}.
Hence, the total increasing number of phrases from such phrases is also $O(\sqrt{n})$.

Overall, the total increasing number of phrases is $O(\sqrt{n})$,
and we obtain $\AS{\edit}(b,n) = O(\sqrt{n})$.

\section{Additive Sensitivity of LZSS and LZ-End}

In this section, we consider the additive sensitivity of the size of LZ-style factorization.
First, we describe LZ-style factorization, excluding LZ78. 
Specifically, they are overlapping LZSS, non-overlapping LZSS,
optimal LZ-End factorization,and greedy LZ-End factorization.

\subsection{Lower bounds}
For any function $c:\Sigma^n \rightarrow \{ 1, \ldots ,n \} $,
the following theorem holds:

\begin{theorem} \label{LZ_lb}
  If $\LZsssr(T) \leq c(T) \leq \LZEnd(T)$ for any string $T \in \Sigma^+$, 
  then $\AS{\edit}(c,n) $ $\in \Omega(n^{\frac{2}{3}})$ holds
  for $\edit \in \{\sub, \ins, \del\}$.
\end{theorem}

\begin{proof}
  Let $T$ be any string of length $n$,
  and $T'$ any string such that $\ed(1,1)=1$.
  By the assumption, $c(T') - c(T) \geq \LZsssr(T') - \LZEnd(T)$ holds. 
  Therefore, to prove $\mathrm{AS}(z,n) \in \Omega(n^{\frac{2}{3}})$ for measure $c$, 
  it is sufficient to show that there exist strings $T,T'$ that satisfy
  $\LZsssr(T') - \LZEnd(T) \in \Omega(n^{\frac{2}{3}})$.

  To demonstrate this, we define the following sequence:
  Let $p \geq 2$ be an integer.
    We define the sequence $\mathsf{Pair}(p)$ of pairs $(l_i,r_i)$ such that
    \[
    (l_{i} ,r_{i}) = 
      \begin{cases}
          (1, 1) & \text{if $i = 1$}, \\
          (l_{i-1}-1,r_{i-1}+1)  & \text{if $l_{i-1} \neq 1$ and $r_{i-1} \neq p$}, \\
          (l_{i-1}+r_{i-1}+1 , 1) & \text{otherwise},
      \end{cases}
    \]
  and $l_i + r_i \leq 2p$ for any $i$. 
  Intuitively, $\mathsf{Pair}(p)$ enumerates all pairs $(l_i, r_i) \in [1,p]^2$ such that $l_i + r_i = k$ in increasing order of $k \in \{2, 3, \ldots, 2p\}$,  
  and for each fixed $k$ in increasing order of $l_i$.
  For instance, $\mathsf{Pair}(2) = (1,1),(1,2),(2,1),(2,2)$.
  The number of elements in $\mathsf{Pair}(p)$ is $p^2$ and thus $\sum_{i=1}^{p^2}(l_i+r_i)=\Theta(p^3)$ holds.

  Let $\Sigma = \{a_1,\ldots,a_p,b_1,\ldots,b_p,x,y,\#_1,\ldots,\#_{p^2} \}$, where $p \geq 2$.
  Furthermore, for each $i \in [1,p]$, we define the strings $A(i) = a_1\cdots a_i$ and $B(i) = b_i \cdots b_1$.
  Then, we consider the following string $T$.
  \begin{eqnarray*} \label{LZSS_T}
  T & = & b_p \cdots b_1x a_1\cdots a_p\#_1 b_1xa_1\#_2 b_2b_1xa_1 \#_3 \cdots \#_i B(l_{i})xA(r_{i}) \#_{i+1} \cdots \#_{p^2}  \\
    & = & B(p)xA(p) \prod_{i=1}^{p^2} \big( \#_i B(l_{i})xA(r_{i}) \big).  
  \end{eqnarray*}
  Observe that $|T| = \Theta(p^3)$.

  The greedy LZ-End factorization of $T$ is as follows.
  \[
    \lzend (T) = b_p|b_{p-1}|\cdots|b_1|x|a_1|a_2|\cdots|a_p|\prod_{i=1}^{p^2} \big( |\#_i|B(l_{i})xA(r_{i})| \big).
  \]
  Since each character in the interval $[1,2p+1]$ is fresh, they become a phrase of length $1$ each.
  Next, for $i \in [1,p^2]$,
  $\#_i B(l_{i})xA(r_{i})$ is partitioned as $|\#_i | B(l_{i})xA(r_{i})|$.
  This is because $\#_i$ is a fresh character,
  $B(l_{i})xA(r_{i})$ has a previous occurrence in $T[p+1-l_{i}..p+1+r_{i}]$,
  and $T[p+1+r_{i}] = a_{r_i}$ is a single-character phrase,
  and the character $\#_{i+1}$ immediately following $B(l_{i})xA(r_{i})$ is fresh.
  Thus, we have $\LZEnd(T) = 2p^2+2p+1$.

  As for substitutions, consider the string
  \begin{eqnarray*} \label{LZSS_sub_T}
  T' & = & b_p \cdots b_1 \underline{y} a_1\cdots a_p\#_1 b_1xa_1\#_2 b_2b_1xa_1 \#_3 \cdots \#_i B(l_{i})xA(r_{i}) \#_{i+1} \cdots \#_{p^2}  \\
    & = & B(p) \underline{y} A(p) \prod_{i=1}^{p^2} \big( \#_i B(l_{i})xA(r_{i}) \big)  
  \end{eqnarray*}
  that is obtained by substituting the first occurrence of $x = T[p+1]$ with a character $y$ (underlined) which does not occur in $T$.

  Let us analyze the structure of overlapping LZSS of $T'$.
  Since each character in the interval $T'[1..2p+1]$ is fresh, they become a phrase of length $1$ each.
  
  For $i = 1$, $\#_1 a_1 x b_1$ is divided into $|\#_1|a_1|x|b_1|$ because $\#_1$ is a fresh character, 
  neither $a_1 x$ nor $x b_1$ has previous occurrences, and $\#_2$ which immediately follows $b_1$ is a fresh character.

  For each $i \in [2,p^2]$, the substring $\#_i B(l_{i})xA(r_{i})$ is divided into $3$ phrases:
  Since $\#_i$ is a fresh character, it becomes a single-character phrase.
  Hence $B(l_{i})xA(r_{i})$ is divided into $2$ phrases,
  and the second phrase ends immediately before $\#_{i+1}$.
  We have two cases:

  \begin{itemize}
    \item When $r_i \neq 1$, $B(l_i)xA(r_{i-1})$ is the longest prefix of $B(l_i)xA(r_i)$ that has a previous occurrence.
  There are two reasons for this. The first is that the occurrence $B(l_i)xA(r_i) = T[p+1-l_{i}..p+1+r_{i}]$ 
  is lost in $T'$ due to the substitution of $x$ with $y$.
  The second is that there exists $j \in [1, i-1]$ such that $l_j = l_i$ and $r_j = r_i-1$. 
  In particular, if $l_i+r_i = k$, then $l_j + r_j = k-1 < k$. 
  Consequently, by the ordering property of $\mathsf{Pair}(p)$, 
  the pair $(l_j, r_j)$ is guaranteed to appear before $(l_i, r_i)$.
  Consequently, $B(l_{i})xA(r_{i})$ is divided into $|B(l_i)xA(r_{i-1})|a_{r_i}|$.

  \item When $r_i = 1$, then $A(r_i) = A(1) = a_1$.
  The previous occurrence of $B(l_i)xa_1 = T[p+1-l_{i}..p+2]$ 
  is also lost in $T'$ by the substitution.
  Then, $B(l_i)$ is the longest prefix of $B(l_i)xa_1$ that has a previous occurrence.
  $B(l_i)$ occurs in $T[p+1-l_{i},p]$, but $B(l_i)x$ does not.
  Consequently, $B(l_i)xa_1$ is divided into $|B(l_i)|xa_1|$.
  \end{itemize}
  Now it follows that for each $i \in [2,p^2]$, the substring $\#_i B(l_{i})xA(r_{i})$ is divided into $3$ phrases.
  Hence, $\LZsssr(T) = 3p^2+2p+2$ holds.  Recall $|T|=n=\Theta(p^3)$.
  As a result, we have $\LZsssr(T') - \LZEnd(T) \geq 3p^2+2p+2 -(2p^2+2p+1) = p^2 \in \Theta(n^{\frac{2}{3}})$.
  
  For insertions and deletions, by considering $T''$ and $T'''$
  obtained from $T$ 
  by inserting $y$ between $x = T[p+1]$ and $b_1 = T[p+2]$ and deleting $x = T[p+1]$ respectively, 
  we can get similar decompositions in the case of substitutions,
  which lead to $\LZsssr(T'') =3p^2 + O(p)$ and $\LZsssr(T''') = 3p^2 + O(p)$ holds.
  This gives us $\LZsssr(T'') - \LZEnd(T) \in \Omega(n^{\frac{2}{3}})$ and $\LZsssr(T''') - \LZEnd(T) \in \Omega(n^{\frac{2}{3}})$.
  
 Overall, $\AS{\sub}(c,n),\AS{\ins}(c,n),\AS{\del}(c,n) \in \Omega(n^{\frac{2}{3}})$ holds.
\end{proof}

Since $\LZsssr(T) \leq \LZss(T) \leq \LZEndOpt(T) \leq \LZEnd(T)$ holds~\cite{KreftN13,KempaS22},
the worst-case additive sensitivity for $\LZsssr(T)$, $\LZss(T)$, $\LZEndOpt(T)$, and $\LZEnd(T)$ are all $\Omega(n^{\frac{2}{3}})$ by Theorem~\ref{LZ_lb}.

\subsection{Upper bounds}
\subsubsection{Overlapping/Non-overlapping LZSS factorization}

Next, we describe the upper bound on the additive sensitivities of LZ-style factorizations.
Specifically, we consider overlapping LZSS, non-overlapping LZSS, and optimal LZ-End.

\begin{theorem} \label{thm:LZSS_ub}
  $\AS{\edit}(\LZsssr,n) = O(n^{\frac{2}{3}})$ and $\AS{\edit}(\LZss,n) = O(n^{\frac{2}{3}})$ hold 
  for $\edit \in \{\sub, \ins, \del\}$.
\end{theorem}

First, we consider the non-overlapping LZSS factorization.
To prove Theorem~\ref{thm:LZSS_ub}, we use the proof of the multiplicative sensitivity of non-overlapping LZSS by Akagi et al~\cite{AKAGI2023104999}.
For the simplicity, we consider the case of substitution (the results for other cases can be shown in a similar way).
Let $T$ be a string and $T'$ the string obtained by editing the $i$-th character of $T$,
$\lzss(T) = f_1 \cdots f_t$, and $\lzss(T') = f'_1 \cdots f'_{t'}$.
We denote the interval of phrase $f_j$ by $[p_j,q_j]$ (i.e., $f_j = T[p_j..q_j]$).
They showed each phrase of $\lzss(T)$ can contain a constant number of beginning positions of phrases of $\lzss(T')$
by considering the three cases.
\begin{proposition}[Theorem~{19} of \cite{AKAGI2023104999}] \label{LZ_ms_ub}
  For $j \in [1,t]$, each interval $[p_j,q_j]$ satisfies as follows.
  \begin{itemize}
    \item[(1)] If $q_j < i$, $[p_j,q_j]$ contains a single beginning position of a phrase in $\lzss(T')$ (i.e., $f'_j = f_j$).
    \item[(2)] If $p_j \leq i \leq q_j$, $[p_j,q_j]$ contains at most constant number of beginning positions of phrases in $\lzss(T')$.
    \item[(3)] If $i < p_j$ and
    \begin{itemize}
      \item[(3-A)] the source of $f_j$ does not contain the position $i$, $[p_j,q_j]$ contains at most one beginning position of phrases in $\lzss(T')$.
      \item[(3-B)] the source of $f_j$ contains the position $i$, $[p_j,q_j]$ contains at most constant number of beginning positions of phrases in $\lzss(T')$.
    \end{itemize}
  \end{itemize} 
\end{proposition}
In case (1) and (3-A), the number of the beginning positions of phrases in $\lzss(T')$ that are contained in $[p_j,q_j]$ does not increase.
In case (2), the number of beginning positions in the interval increases by a constant number, but such phrase is unique in $\lzss(T)$.
Here, we discuss an upper bound of the total increasing number of beginning positions of phrases in case (3-B) in terms of $n$.
First, we consider phrases in case (3) of length at least $n^{\frac{1}{3}}$.
Then, such phrases can occur at most $O(n^{\frac{2}{3}})$ times in $T$, since $|T|= n$.
Hence the total increasing number of beginning positions of phrases from such phrases is $O(n^{\frac{2}{3}})$, 
since each phrase can contain a constant number of such positions by Proposition~\ref{LZ_ms_ub}.
Next, we consider the number of phrases whose length is smaller than $n^{\frac{1}{3}}$.
If there exists $f_{j'}$ in case (3) such that $j' < j$ and $f_j = f_{j'}$, 
then, $f_{j'} = T'[p_{j'}..q_{j'}]$ can be a source of $T'[p_j..q_j]$.
Thus, in this case, the number of beginning positions of phrases in the interval $[p_{j},q_{j}]$ does not increase.
On the other hand, the maximum number of distinct phrases in case (3) is $O(n^{\frac{2}{3}})$
since the number of substrings of length at most $k$ that contain a position $i$ is $O(k^2)$.
For such $O(n^{\frac{2}{3}})$ phrases, $[p_j,q_j]$ contains at most a constant number of beginning positions of phrases in $\lzss(T')$ by Proposition~\ref{LZ_ms_ub}.
Hence, the total increasing number of beginning positions of phrases from such phrases is also $O(n^{\frac{2}{3}})$.

Overall, the total increasing number of phrases is $O(n^\frac{2}{3})$,
and we obtain $\AS{\edit}(\LZsssr,n) = O(n^{\frac{2}{3}})$.

Akagi et al. proved the multiplicative sensitivity of overlapping LZSS in a similar way~\cite{AKAGI2023104999}.
We can also obtain $\AS{\edit}(\LZss,n) = O(n^{\frac{2}{3}})$ in a similar way.

\subsubsection{Optimal LZ-End factorization.}
As for the sensitivity of optimal LZ-End factorization, we can obtain the following upper bounds.
Together with Theorem~\ref{LZ_lb}, we can also obtain a tight bound $\Theta(n^{\frac{2}{3}})$ for $\AS{\edit}(\LZEndOpt, n)$.

\begin{theorem} \label{LZEnd_ub}
  $\AS{\edit}(\LZEndOpt,n) = O(n^{\frac{2}{3}})$ holds 
  for $\edit \in \{\sub, \ins, \del\}$.
\end{theorem}

Similarly to the case of LZSS (Theorem~\ref{LZEnd_ub}) in the previous section,  
we also use an idea for the multiplicative sensitivity of optimal LZ-End factorization.
However, there is no result for it.
Here, we first show an upper bound of the multiplicative sensitivity of the optimal LZ-End factorization by a similar idea.
For ease of discussion, we consider the case of substitution.
Let $T$ be a string and $T'$ the string obtained by substituting the $i$-th character $T[i] = a$ with $b$, and
$f_1, \ldots, f_t$ an LZ-End factorization of $T$.
We denote the interval of phrase $f_j$ by $[p_j,q_j]$ (i.e., $f_j = T[p_j..q_j]$).
Moreover, let $I$ be an integer that satisfies $p_I \leq i \leq q_I$ (i.e., the index of a phrase that contains the edited position $i$).
We consider the following operations to $f_1, \ldots, f_t$.
\begin{proposition} \label{LZEnd_ms_ub}
  An LZ-End factorization of $T'$ can be obtained
  by applying the following operations to $f_1, \ldots, f_t$.
  \begin{itemize}
    \item[(1)] If $q_j < i$, we don't need any changes.
    \item[(2)] If $p_j \leq i \leq q_j$, we divide the phrase $f_j$ into at most $I+1$ phrases. 
    \item[(3)] If $i < p_j$ and 
    \begin{itemize}
      \item[(3-A)] the source of $f_j$ does not contain the position $i$, we don't need any changes.
      \item[(3-B)] the source of $f_j$ contains the position $i$, we divide the phrases $f_j$ into at most three phrases. 
    \end{itemize}
  \end{itemize} 
\end{proposition}

We show how to construct an LZ-End factorization of $T'$ from $f_1, \ldots, f_t$.
In case (1), for the first $I-1$ phrases, 
$T'$ has the same LZ-End factorization since the prefix is unchanged by the substitution.
In case (2), let $T[p_I..q_I] = f_I = w_1aw_2$ and $T'[p_I..q_I] = w_1bw_2$. 
First, the substring $b$ can become a phrase of LZ-End of $T'$ since $b$ is a single character.
Next, since $w_1aw_2$ is a phrase of an LZ-End of $T$, there exists an integer $k < I$ such that $w_1aw_2$ is a suffix of $f_1 \cdots f_k$.
Then, $w_2$ in the source of $f_I$ can still be a source of a new phrase $w_2$ in $T'$.
In the rest of this part, we show that $w_1 = T[p_I..p_I+|w_1|-1]$ can be divided into at most $I$ LZ-End phrases of $T'$.
Since $w_1$ is a prefix of $f_I$, $w_1$ has a previous occurrence in its source.
By traversing phrases to the left, we can find a phrase $f_{k'}$ 
such that $k' < I$ and the LZ-End boundary $q_{k'}$ contained in an occurrence of $w_1$.
Assume that $k'$ is the largest integer in such indices.
Then $w_1$ at the position can be written as $w_1 = uv$,
where $u$ is a suffix of $f_1 \cdots f_{k'}$ and $v$ is a prefix of $f_{k'+1} \cdots f_{I-1}$.
From this structure, prefix $u$ of $T'[p_I..q_I]$ can be an LZ-End phrase such that the source ends at $q_{k'}$.
For the remaining suffix $v$ of $w_1$, in a similar argument, 
we can find a phrase $f_{k''}$ 
such that $k'' < k'+1~(< I)$ and the LZ-End boundary $q_{k''}$ contained in an occurrence of $v$.
Then $v$ at the position can be written as $v = u'v'$, and
the substring $u'$ of $T'[p_I..q_I]$ can be an LZ-End phrase such that the source ends at $q_{k''}$.
By traversing phrases and dividing $w_1$ recursively, 
we can finally obtained at most $I-1$ LZ-End phrases of $T'$ 
such that the concatenation of these phrases is $T'[p_I..p_I+|w_1|-1] = w_1$
(see Fig.~\ref{Fig_LZ-End_case2} for an illustration). 
Thus we can obtain at most $I+1$ LZ-End phrases of $T'$ from $f_I$.
We consider phrases in case (3-A).
$f_j$ has a source $T[p_j'..q_j']$.
Since $i \notin [p_j',q_j']$, this string also appear at the same position in $T'$.
Our construction of an LZ-End factorization of $T'$ only use divisions of phrases of the original factorization of $T$.
This implies that there exists a phrase of an LZ-End factorization of $T'$ that ends at $q_j'$.
Hence, $T'[p_j..q_j]$ can be a phrase whose source is $T'[p_j'..q_j']$.
In case (3-B), let $f_{j} = w_3aw_4$ such that $a$ corresponds to the edited position in its source.
We can see that $w_3$, $a$, $w_4$ can be LZ-End phrases in $T'$,
because $w_3$ occurs at $[p_j',i-1]$ and there exists an LZ-End phrase in $T'$ that ends at $i-1$ by case (2),
$a$ is a single character, and
$w_4$ still occurs at $[i+1,q_j']$ in $T'$ (as a suffix of the original phrase in $T$).

Overall, we can obtain an LZ-End factorization of $T'$ from a given LZ-End factorization of $T$ 
by dividing each phrase in $T$ into at most three parts.

\begin{figure}[t]
  \centering
  \includegraphics[keepaspectratio, width=12cm]{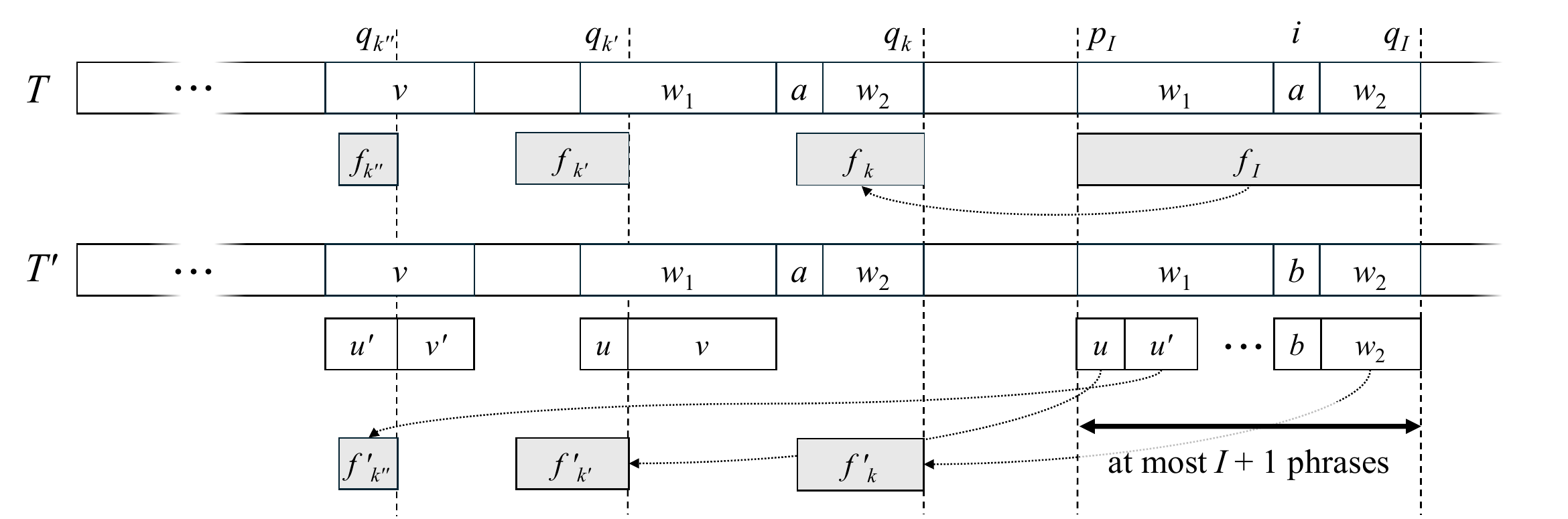}
  \caption{Illustration of the division operation for Case~(2).}
  \label{Fig_LZ-End_case2}
\end{figure} 

Moreover, respect to deletions, a proposition that is completely the same Proposition~\ref{LZEnd_ms_ub} holds.
As for insertions, although it is almost the same as Proposition~\ref{LZEnd_ms_ub},
in case (3-B), $f_j$ can be divided into at most two phrases not three.
This can be proved the similar way of Akagi et al. [Theorem~{18} of \cite{AKAGI2023104999}].
Then, we get the following upper bound for the multiplicative sensitivity of $\LZEndOpt$.

\begin{corollary}
  $\MS{\sub}(\LZEndOpt,n) = 3$ ,$\MS{\del}(\LZEndOpt,n) = 3 , \MS{\ins}(\LZEndOpt,n) = 2$.
\end{corollary}

Then, we describe the upper bound of the additive sensitivity of the optimal LZ-End factorization in a similar way of Theorem~\ref{thm:LZSS_ub}.
From Proposition~\ref{LZEnd_ms_ub}, 
in case (1) or (3-A), the number of phrases remains unchanged.
Here, we give an upper bound of increasing number of phrases in case (2) and (3-B) in terms of $n$. 
If $I > n^{2/3}$, then the number of phrases increases by at most $O(n^{\frac{2}{3}})$, since $|T|= n$.
If $I < n^{2/3}$, then the number of phrases in the interval $[p_I,q_I]$ (case~(2)) 
increases at most $I+1 \in O(n^{\frac{2}{3}})$. Next, by the same way with non-overlapping LZSS, 
we can prove that the upper bound of increasing number of phrases is $O(n^{\frac{2}{3}})$.
In any case, the total increasing number of phrases is $O(n^{\frac{2}{3}})$.
and we obtain $\AS{\edit}(\LZEndOpt,n) = O(n^{\frac{2}{3}})$ for $\edit \in \{\sub, \ins, \del\}$.

\section{Additive Sensitivity of LZ78}

In this section, 
we present the tight bound
$\AS{\edit}(\LZSevenEight,n) = \Theta(n)$
for the additive sensitivity of the Lempel-Ziv 78 factorization size $\LZSevenEight$ for all types of edit operations $\edit \in \{\sub, \ins, \del\}$.
Note that the upper bound $\AS{\edit}(\LZSevenEight,n) = O(n)$ trivially holds,
since we have $1 \leq \LZSevenEight(T) \leq n$ for any string $T$ of length $n$.

We show the following lower bound:
\begin{theorem} \label{LZ78_main_result}
$\AS{\edit}(\LZSevenEight,n) = \Omega(n)$ holds for $\edit \in \{\sub, \ins, \del\}$.
\end{theorem}

Due to the space constraint, we just present a family of strings which we will use in our proof.
A full proof can be found in Appendix.
Let $p \in \mathbb{N}$ and $\Sigma = \{ a_1,..,a_p,b_1,..,b_p,c_1,..,c_p, \# \}$.
Consider the string $T$ of length $7p$ and $T'$ (for substitutions) 
that can be obtained by substituting $T[4m+1] = a_1 $ with a character $\#$ which does not occur in $T$.
\[
  T = \prod_{i=1}^{p} \big(c_i\big) \prod_{i=1}^{p} \big(a_i a_i b_i\big) \prod_{i=1}^{p}  \big(a_i b_i c_i\big),~
  T' = \prod_{i=1}^{p} \big(c_i\big) \prod_{i=1}^{p} \big(a_i a_i b_i\big)   \# b_1 c_1 \prod_{i=2}^{p}  \big(a_i b_i c_i\big).
\]
For these strings, we can show $\LZSevenEight(T) = 4p$ and $\LZSevenEight(T') = 5p+1$.
Hence, we get $\AS{\sub}(\LZSevenEight,n) \geq p+1 \in \Omega(n)$ because $n = 7p$.
In the case of insertions and deletions, with the strings
obtained from $T$ 
by inserting $\#$ between $a_1 = T[4p+1]$ and $b_1 = T[4p+2]$ or
by deleting $a_1 = T[4p+1]$ respectively, 
we get similar decompositions as in the case of substitutions,
which give us $\AS{\ins}(\LZSevenEight,n), \AS{\del}(\LZSevenEight,n) \in \Omega(n)$.

\clearpage

\bibliographystyle{splncs04}
\bibliography{ref}

\begin{thebibliography}{10}
\providecommand{\url}[1]{\texttt{#1}}
\providecommand{\urlprefix}{URL }
\providecommand{\doi}[1]{https://doi.org/#1}

\bibitem{AKAGI2023104999}
Akagi, T., Funakoshi, M., Inenaga, S.: Sensitivity of string compressors and
  repetitiveness measures. Information and Computation  \textbf{291},  104999
  (2023). \doi{10.1016/j.ic.2022.104999},
  \url{https://www.sciencedirect.com/science/article/pii/S0890540122001547}

\bibitem{BlockiLSY25}
Blocki, J., Lee, S., Garcia, B.S.Y.: Differentially private compression and the
  sensitivity of {LZ77}. CoRR  \textbf{abs/2502.09584} (2025).
  \doi{10.48550/ARXIV.2502.09584},
  \url{https://doi.org/10.48550/arXiv.2502.09584}

\bibitem{Burrows94ablock-sorting}
Burrows, M., Wheeler, D.: A block-sorting lossless data compression algorithm.
  Tech. rep., DIGITAL SRC RESEARCH REPORT (1994)

\bibitem{CharikarLLPPSS05}
Charikar, M., Lehman, E., Liu, D., Panigrahy, R., Prabhakaran, M., Sahai, A.,
  Shelat, A.: The smallest grammar problem. {IEEE} Trans. Inf. Theory
  \textbf{51}(7),  2554--2576 (2005)

\bibitem{GiulianiILRSU25}
Giuliani, S., Inenaga, S., Lipt{\'{a}}k, Z., Romana, G., Sciortino, M., Urbina,
  C.: Bit catastrophes for the burrows-wheeler transform. Theory Comput. Syst.
  \textbf{69}(2), ~19 (2025). \doi{10.1007/S00224-024-10212-9},
  \url{https://doi.org/10.1007/s00224-024-10212-9}

\bibitem{KempaP18}
Kempa, D., Prezza, N.: At the roots of dictionary compression: string
  attractors. In: {STOC} 2018. pp. 827--840. {ACM} (2018)

\bibitem{KempaS22}
Kempa, D., Saha, B.: An upper bound and linear-space queries on the lz-end
  parsing. In: Proceedings of the 2022 {ACM-SIAM} Symposium on Discrete
  Algorithms, {SODA} 2022. pp. 2847--2866. {SIAM} (2022)

\bibitem{KociumakaNP20}
Kociumaka, T., Navarro, G., Prezza, N.: Towards a definitive measure of
  repetitiveness. In: {LATIN}. pp. 207--219 (2020)

\bibitem{KreftN13}
Kreft, S., Navarro, G.: On compressing and indexing repetitive sequences.
  Theor. Comput. Sci.  \textbf{483},  115--133 (2013)

\bibitem{repetitiveness_Navarro21a}
Navarro, G.: Indexing highly repetitive string collections, part {I:}
  repetitiveness measures. {ACM} Comput. Surv.  \textbf{54}(2),  29:1--29:31
  (2021). \doi{10.1145/3434399}, \url{https://doi.org/10.1145/3434399}

\bibitem{Rytter03}
Rytter, W.: Application of {Lempel-Ziv} factorization to the approximation of
  grammar-based compression. Theor. Comput. Sci.  \textbf{302}(1-3),  211--222
  (2003)

\bibitem{StorerS82}
Storer, J.A., Szymanski, T.G.: Data compression via textual substitution. J.
  {ACM}  \textbf{29}(4),  928--951 (1982)

\bibitem{LZ77}
Ziv, J., Lempel, A.: A universal algorithm for sequential data compression.
  IEEE Transactions on Information Theory  \textbf{23}(3),  337--343 (1977)

\bibitem{LZ78}
Ziv, J., Lempel, A.: Compression of individual sequences via variable-rate
  coding. {IEEE} Trans. Inf. Theory  \textbf{24}(5),  530--536 (1978)

\end{thebibliography}

\clearpage

\section*{Appendix}

\subsection*{Proof of Theorem~\ref{LZ78_main_result}}

  Let $p \in \mathbb{N}$ and $ \Sigma = \{ a_1,..,a_p,b_1,..,b_p,c_1,..,c_p, \# \}. $
  Consider the following string
  \begin{align*}
    T & = c_1c_2 \cdots c_p a_1 a_1 b_1 a_2 a_2 b_2 \cdots a_p a_p b_p a_1 b_1 c_1 a_2 b_2 c_2 \cdots a_p b_p c_p \\
    & = \prod_{i=1}^{p} \big(c_i\big) \prod_{i=1}^{p} \big(a_i a_i b_i\big) \prod_{i=1}^{p}  \big(a_i b_i c_i\big).
  \end{align*}
  Clearly, $|T| = n = 7p$ holds. 
  The LZ78 factorization of $T$ is as shown below:
  \begin{align*}
    \lzseveneight(T) & = |c_1|c_2| \cdots |c_p| a_1 |a_1 b_1| a_2| a_2 b_2| \cdots |a_p| a_p b_p | a_1  b_1 c_1 | a_2  b_2 c_2 | \cdots | a_p b_p c_p |,
  \end{align*}
  where each $|$ denotes a phrase boundary.
  To see why,
  observe that each $c_i$ is a fresh character
  in the prefix interval $[1,p]$ in $T$,
  and thus $c_i$~($1 \leq i \leq p$) is a phrase.
  Next, in the interval $[p+1,4p]$,
  $a_ia_ib_i$ for each $1 \leq i \leq p$ is divided into $|a_i|a_ib_i|$.
  This is because the left $a_i$ is a fresh character and it becomes the phrase of length $1$.
  Also, the substring $a_ib_i$ has no previous occurrence as a phrase. 
  Lastly, in the interval $[4p+1 ,7p]$,
  the substring $a_ib_ic_i$ for each $1 \leq i \leq p$ becomes a phrase.
  This is because $a_ib_i$ has a previous occurrence as a phrase, but $a_ib_ic_i$ does not.
  Hence, we have $\LZSevenEight(T) = 4p$.

  As for substitutions, consider the following string $T'$
  \begin{align*}
    T' & = c_1c_2 \cdots c_p a_1 a_1 b_1 a_2 a_2 b_2 \cdots a_p a_p b_p \# b_1 c_1 a_2 b_2 c_2 \cdots a_p b_p c_p \\
    & = \prod_{i=1}^{p} \big(c_i\big) \prod_{i=1}^{p} \big(a_i a_i b_i\big)   \# b_1 c_1 \prod_{i=2}^{p}  \big(a_i b_i c_i\big)
  \end{align*} 
  which can obtained by substituting $T[4m+1] = a_1 $ with a character $\#$ which does not occur in $T$.
  Let us analyze the structure of LZ78 factorization of $T'$. First, the factorization of $T'[1..4p]$ is the same as $T[1..4p]$ before the edit.
  Then, since $\# = T'[4p+1]$ is a new character and $b_1 = T'[4p+2]$ has no previous occurrence as a phrase,
  $\#b_1$ is factorized into $|\#|b_i|$. Next, the substring $c_1a_2$ has no previous occurrence (as a phrase), but $c_1$ has.
  Therefore, $c_1a_2$ becomes a phrase.
  Then, $c_{i-1}a_ib_i$ for each $2 \leq i \leq p$ 
  is divided into $|c_{i-1}a_i|b_i|$, since $c_{i-1}$ has a previous occurrence as a phrase, but $c_{i-1}a_i$ does not.
  Thus, the LZ78 factorization of $T'$ is 
  \begin{align*}
    \lzseveneight(T') & = |c_1| \cdots |c_p| a_1 |a_1 b_1| \cdots |a_p| a_p b_p |
     \#| b_1 | c_1 a_2 | b_2 | c_2 a_3 |\cdots |c_{p-1}a_p |b_p | c_p |
  \end{align*}
  which contains $\LZSevenEight(T') = 5p+1$ phrases. Hence, we get $\AS{\sub}(\LZSevenEight,n) \geq \LZSevenEight(T') - \LZSevenEight(T) = 5p+1 - 4p = p+1 \in \Omega(n)$ because $n = 7p$.
  
  
  In the case of insertions and deletions, with the strings
  obtained from $T$ 
  by inserting $\#$ between $a_1 = T[4p+1]$ and $b_1 = T[4p+2]$ or
  by deleting $a_1 = T[4p+1]$ respectively, 
  we get similar decompositions as in the case of substitutions,
  which give us $\AS{\ins}(\LZSevenEight,n), \AS{\del}(\LZSevenEight,n) \in \Omega(n)$.

\end{document}